\documentclass{article}
\usepackage{amsmath, amssymb, amsfonts, amsthm}
\usepackage{graphicx}
\usepackage{hyperref}
\usepackage{tikz}
\usepackage{tikz-qtree}
\usepackage[margin=2.5cm]{geometry}
\usepackage{yhmath}
\usepackage{mathdots}
\usepackage{MnSymbol}
\usepackage{fancyhdr}
\newcommand{\Free}[1]{\ensuremath{\text{\textsl{Free}}(#1)}}
\newcommand{\Forb}[1]{\ensuremath{\text{\textsl{Forb}}(#1)}}
\newcommand{\Prob}[1]{\ensuremath{\text{\textsc{#1}}}}

\newcommand{\classX}[1]{\ensuremath{\text{\textsf{\textbf{#1}}}}} 
\newcommand{\classP}{\classX{P}}
\newcommand{\classNP}{\classX{NP}}
\newcommand{\NPC}{\classX{NP-complete}}

\newcommand{\NPH}{\classX{NP-hard}}
\newcommand{\etal}{\textit{et al.}}

\newtheorem{theorem}{Theorem}

\newtheorem{corollary}{Corollary}

\begin{document}

\title{Reducibility among NP-Hard graph problems and boundary classes}
\author{Syed Mujtaba Hassan\footnote{\url{mujtaba.hassan@sse.habib.edu.pk}, Computer Science Department, Habib University, Karachi, Pakistan (corresponding author)} \and Shahid Hussain\footnote{\url{shahidhussain@iba.edu.pk}, Computer Science Department, Institute of Business Administration, Karachi, Pakistan}
\and Abdul Samad \footnote{\url{abdul.samad@sse.habib.edu.pk}, Computer Science Department, Habib University, Karachi, Pakistan}}
\date{July 9, 2025}
\maketitle

\begin{abstract}
    Many $\NPH$ graph problems become easy for some classes of graphs. For example, \Prob{coloring} is easy for bipartite graphs, but $\NPH$ in general. 
    So we can ask question like when does a hard problem become easy? 
    What is the minimum substructure for which the problem remains hard? We use the notion of \emph{boundary classes} to study such questions.
    In this paper, we introduce a method for transforming the boundary class of one $\NPH$ graph problem into a boundary class for another problem.
    If $\Pi$ and $\Gamma$ are two $\NPH$ graph problems where $\Pi$ is reducible to $\Gamma$, we transform a boundary class of $\Pi$ into a boundary class of $\Gamma$. 
    More formally if $\Pi$ is reducible to $\Gamma$, where the reduction satisfies certain conditions, then $X$ is a boundary class of $\Pi$ if and only if the image of $X$ under the reduction is a boundary class of $\Gamma$. 
    This gives us a relationship between boundary classes and reducibility among several $\NPH$ problems. To show the strength of our main result, we apply our theorem to obtain some previously unknown boundary classes for a few graph problems namely;
    \Prob{vertex-cover}, \Prob{clique}, \Prob{traveling-salesperson}, \Prob{bounded-degree-spanning-tree}, \sloppy\Prob{subgraph-isomorphism} and \Prob{clique-cover}. \smallskip

\end{abstract}

\renewcommand\thefootnote{}

\renewcommand\thefootnote{\fnsymbol{footnote}}
\setcounter{footnote}{1}
\section{Introduction}
In computational complexity theory we are interested in classifying computational problems based on their inherent difficulty i.e., how hard or easy is it to answer a decision problem.
Many graph problems are $\NPH$ that is every problem in $\classNP$ is reducible to them \cite{Sipser13}. We say that these problems are at least as hard as the hardest problems in class $\classNP$.
However, we see that many of these $\NPH$ graph problems become tractable when the class of input graphs is restricted \cite{Alekseev07}.
The natural question is when are these hard problems easy and what is the smallest substructure that causes these problems to be hard?
Boundary classes are used to study exactly this. The notion of boundary classes was first introduced in \cite{Alekseev03} by Vladimir Alekseev to study the maximum independent set problem. 
The idea is to draw a boundary around the hard instances of the problem separating the easy and hard classes. 

Over the years many boundary classes have been proved for many $\NPH$ graph problems such as \emph{independent-set} \cite{Alekseev03,Alekseev07}, \emph{vertex-coloring} \cite{Lozin11,Malyshev11,Malyshev13-2}, \emph{edge-coloring} \cite{Malyshev13-2,Malyshev14}, 
\emph{hamiltonian-cycle} \cite{Korpelainen18,Lozin11,Nicholas10,Munaro17}, \emph{hamilton}\emph{ian-path}~\cite{Munaro17}, \emph{dominating-set}~\cite{Alekseev07,Alekseev04,Dakhno23,Malyshev16}, \emph{upper-domination}~\cite{Hussain16}, \emph{feedback-vertex-set}~\cite{Munaro17}, 
\emph{role-coloring}~\cite{purcell21}, \emph{coupon-coloring}~\cite{purcell21}, \emph{independent-dominating-set}~\cite{Alekseev07}, \emph{induced-matching}\cite{Alekseev07}, \emph{dissciation-set}~\cite{Alekseev07}, \emph{longest-cycle}~\cite{Alekseev07}, \emph{longest-path}~\cite{Alekseev07}, 
\emph{edge-domination}~\cite{Alekseev07}, $P_3$\emph{-factor}\cite{Alekseev07}, \emph{maximum-induced-subgraph}~\cite{Malyshev14-2}, $k$\emph{-path-partition}\cite{Korpelainen18}, \emph{efficie}\emph{nt-edge-domination}~\cite{Domingos11,Korpelainen18}, and \emph{edge-ranking}~\cite{Malyshev14-3}.

One question is to find a relationship between existing boundary classes. In this paper, we introduced a relationship between boundary classes of $\NPH$ graph problems which are reducible to each other.
We introduce a method for transforming a boundary class of one $\NPH$ graph problem into a boundary class for another problem when a suitable reduction between them exists.
If $\Pi$ and $\Gamma$ are two $\NPH$ graph problems where $\Pi$ is reducible to $\Gamma$, we transform a boundary class of $\Pi$ into a boundary class of $\Gamma$. More specifically we show that if $\Pi$ is reducible to $\Gamma$, where the reduction satisfies certain conditions (described in detail in section~\ref{sec: main result}) then $X$ is a boundary class of $\Pi$ if and only if the image of $X$ under the reduction is a boundary class of $\Gamma$. We introduce this theorem more formally in Section~\ref{sec: main result} of this paper. This gives us a relationship between boundary classes and reducibility amongst $\NPH$ problems. This theorem also gives us a method to find new boundary classes for one $\NPH$ graph problem from a boundary class of another $\NPH$ graph problem given there is a reduction between them that satisfies the conditions.

We apply our theorem to obtain some previously unknown boundary classes for the following graph problems \emph{vertex-cover}, \emph{clique}, \emph{traveling-salesperson}, \emph{bounded-degree-spanning-tree}, \emph{clique-cover} and \emph{subgraph-isomorphism}. For all these problems these are the first boundary classes. This shows the strength of our result. 
Given an $\NPH$ graph problem $\Pi$, with a boundary class $X$ we can find a boundary class for another $\NPH$ graph problem $\Gamma$ we just need to find a reduction between them (if such a reduction exists) that satisfies the conditions of our theorem. 
Thus, our result provides a framework for systematically discovering new boundary classes for many $\NPH$ graph problems which are currently unknown.

\section{Preliminaries}
An undirected graph is a pair $G = (V, E)$ of sets $V$ and $E$ such that the elements of $E$ are $2$-element subsets of $V$. An edge $e$ is represented as $\{u, v\}$ such that~$u, v\in V$ and $u\neq v$. We sometimes use $V(G)$ and $E(G)$ to represent the vertex and edge sets of the graph $G$, respectively. We say a graph $H$ is an \emph{induced subgraph} of the graph $G$ if $H$ is obtained by removing some vertices of $G$ and keeping the edges between the remaining vertices. We represent an induced subgraph $H$ of $G$ as $H\subseteq G$. For some subset of vertices $U\subseteq V$, we write $G-U$ to denote the subgraph of $G$ obtained by removing the vertices in~$U$. We represent the complement of a graph $G$ as $\overline{G}$. The complement of a graph $G$ is another graph where the vertices are the same as $G$ and two vertices are adjacent in $\overline{G}$ if and only if they are not adjacent in $G$. The degree of a vertex $v$ in $G$ is the number of edges incident on $v$ and is represented as $d_G(v)$. We use $\Delta(G)$ and $\delta(G)$ to represent the maximum and minimum degrees of $G$, respectively. A \emph{path} $P = (V', E')$ is a non-empty subgraph of a graph $G = (V, E)$ such that $V' = \{v_1, v_2, \ldots, v_k\}$ is a set of distinct vertices and $E' = \{ \{v_1, v_2\}, \{v_2, v_3\}, \ldots, \{v_{k-1}, v_{k}\} \}$. A path of length $10$ can thus be denoted as $P_{10}$. We define a \emph{cycle} as a path of at least three vertices (at least $P_3$) such that $E'$ contains $\{v_1, v_k\}$.

We define the \emph{chromatic number} of a graph $G$ as $\chi(G)$. The chromatic number is the minimum number of colors needed to color each vertex of the graph such that no two adjacent vertices are of the same color. 

A graph $G$ is called a \emph{complete graph} $K_n$ of $n\in\mathbb{N}$\footnote{the set $\mathbb{N}$ is the set of natural numbers such that $\mathbb{N} = \{0, 1, 2, \ldots\}$} vertices if and only if each vertex in $G$ is connected with every other vertex of $G$. A complete subgraph of a graph is called a \emph{clique}. We define the clique size $\omega(G)$ of a graph $G$ as the size of a maximum clique in $G$. A \emph{clique cover} of a graph $G$ is a partition of vertices of $G$ into cliques. For example, a clique cover of a graph $G$ of size $k$ is the partition $G$ into $k$ cliques.

We say a subset of vertices $S\subseteq V$ of a graph $G$ is called an \emph{independent set} if no two vertices in $S$ are adjacent. The size of a maximum independent set in a graph $G$ is called the \emph{independence number} of $G$ and is represented as $\alpha(G)$.

We define the \emph{vertex cover} of a graph $G$ as a subset of vertices $S\subseteq V$ such that each edge in $G$ is incident on at least one vertex in $S$. The size of a minimum vertex cover in a graph $G$ is called the \emph{vertex cover number} of $G$ and is represented as $\beta(G)$.

We say an undirected graph $G$ is a tree if and only if $G$ is connected and has no cycles. In other words, a connected undirected graph $G$ with $n$ vertices and $n-1$ edges is a tree.

A \emph{spanning tree} of a graph $G$ is a subgraph of $G$ that is a tree and contains all the vertices of $G$.  A \emph{$k$-bounded degree spanning tree} of a graph $G$ is a spanning tree of $G$ such that each vertex in the tree has a degree at most $k$.

A path $P$ in a graph $G$ is called a \emph{Hamiltonian path} if $V(P) = V(G)$. A cycle $C$ in a graph $G$ is called a \emph{Hamiltonian cycle} if $V(C) = V(G)$.

We say that two graphs $G_1$ and $G_2$ are \emph{isomorphic} if and only if there exists a bijection $f: V(G_1) \to V(G_2)$ such that $\{u, v\}\in E(G_1)$ if and only if $\{f(u), f(v)\}\in E(G_2)$.

A collection of graphs with some shared property is called a \emph{graph class}. For example, the class of all graphs with chromatic number $3$ is a graph class. We say a graph class $X$ is \emph{hereditary} if and only if for every graph $G\in X$, every induced subgraph of $G$ is also in $X$. 

For a class of graphs $Z$, we say the collection $\Free{Z}$ represents the class of all graphs that do not contain any induced subgraphs isomorphic to a graph in~$Z$. For example, $\Free{K_5}$ represents the class of all graphs that do not contain any induced subgraph isomorphic to $K_5$. A class of graphs $X$ is called \emph{hereditary} if and only if there exists a set of graphs $Z$ such that $X = \Free{Z}$ \cite{Alekseev07}. If $X = \Free{Z}$ and the class $Z$ is minimal and unique then we say that the class $Z$ is the \emph{forbidden class} of $X$, denoted as $\Forb{X}$.

In this paper, we will define various graph theory-related computational problems. These problems come in two flavors: \emph{decision problems} and \emph{optimization problems}. For example questions like ``Does the graph $G$ contain a clique of size $k$?'' or ``What is the size of the maximum clique in the graph $G$?'' are decision and optimization problems, respectively. We will define these problems formally. For example, we can define the decision problem of a clique as follows:
\[
\Prob{clique} = \{ \langle G, k \rangle \mid G \text{ contains a clique of size } k \}.
\]
NThis set $\Prob{clique}$ is infinite, consisting of all pairs $\langle G, k \rangle$ where $G$ contains a clique of size $k$. So the question ``Does $G$ contain a clique of size $k$?'' can be answered by checking if $\langle G, k \rangle \in \Prob{clique}$. 

In computational complexity theory, we often solve one problem by means of another problem, this is called reduction. For two problems $A$ and $B$ defined over some alphabet $\Sigma$, we say that the problem $A$ reduces to the problem $B$ if and only if there exists a computable function $f: \Sigma^* \to \Sigma^*$ such that for every $x\in \Sigma^*$ we have $x\in A$ if and only if $f(x)\in B$. We denote this as $A\leq B$. Furthermore, if the function can be computed in \emph{polynomial time} we say that reduction is \emph{polynomial time reduction} and denote it as $A\leq_\classP B$~\cite{Karp72,Sipser13}. 

We say that the class \classP\ contains all those decision problems that can be solved in polynomial time on a deterministic machine \cite{Karp72,Sipser13}. Similarly, the class \classNP\ contains all those decision problems that can be solved in polynomial time on a non-deterministic machine. Interestingly, the class \classNP\ can also be defined as the class of all decision problems for which a solution can be verified in polynomial time with a polynomial size certificate \cite{Karp72,Sipser13}. A decision problem $A$ is said to be \NPH\ if and only if for every $B\in \classNP$ we have $B\leq_\classP A$. 

\section{Boundary Classes}
In this section, we define the notion of boundary classes. Boundary classes were first introduced by Alekseev in 2003 \cite{Alekseev03} to study the independent set problem. The notion was later provided in the most general form by Alekseev \etal in 2007 \cite{Alekseev07}.
The boundary class concept separates or defines the boundary between the easy and hard problems of graph theory. For an \NPH\ graph problem $\Pi$, a hereditary class of graphs $X$ is called $\Pi$-hard if $\Pi$ remains \NPH\ when the input is restricted $X$, and $X$ is called $\Pi$-easy if $\Pi$ becomes polynomial-time solvable when the input is restricted to $X$ \cite{Alekseev07}. 
We say $X$ is a $Y$-\emph{limit class} of the problem $\Pi$ (or $X$ is $(\Pi, Y)$-limit for short) for $\Pi$-hard class $Y$, if there exists a sequence of $\Pi$-hard classes $X_1\supseteq X_2\supseteq \ldots$ of $Y$ such that $X = \bigcap_{i\in \mathbb{Z}^+} X_i$ \cite{Alekseev07}. We will call $X$ a limit class for $\Pi$ if there is a $\Pi$-hard class $Y$ such that $X$ is $(\Pi , Y)$-limit. 
The minimal limit class under set inclusion is called a \emph{boundary class} with respect to $Y$ (or $(\Pi, Y)$-boundary for short) \cite{Alekseev07}. We will call $X$ a boundary class for $\Pi$ if there is a $\Pi$-hard class $Y$ such that $X$ is $(\Pi , Y )$-boundary. 
The following two theorems by Alekseev \etal shows the strength of the notion of boundary classes in studying the computational complexity of graph problems. 

\begin{theorem}[\cite{Alekseev07}]\label{theorem1}
    A problem $\Pi$ is $\NPH$ in a class $Y$ of graphs if and only if $Y$ contains a $(\Pi, Y)$-boundary class.
\end{theorem}

\begin{theorem}[\cite{Alekseev07}]\label{theorem2}
    A problem $\Pi$ is $\NPH$ in a subclass $Z \subset Y$ defined by finitely many forbidden induced subgraphs with respect to $Y$ if and only if $Z$ contains a $(\Pi, Y)$-boundary class.
\end{theorem}

This shows us that if we obtain all the boundary classes for an $\NPH$ problem $\Pi$ we can then completely classify finitely defined graph classes with respect to the complexity of $\Pi$. This makes the notion of boundary classes very powerful in studying the computational complexity of $\NPH$ graph problems.

\section{Problems}\label{sec:problems}
In this section, we provide a formal definition various graph theoretic computational problems that we will use in later results.  

\begin{eqnarray}
    \Prob{independent-set} &=& \{ \langle G, k \rangle \mid \alpha(G) = k \} \\
    \Prob{vertex-coloring} &=&\{\langle G, k \rangle \mid \chi(G) \leq k \} \\
    % \Prob{edge-coloring} &=& \{ \langle G, k \rangle \mid \chi'(G) \leq k \} \\
    \Prob{hamiltonian-path} &=& \{\langle G \rangle \mid \text{there exists a path $P\subseteq G$ } \nonumber \\
    && \text{such that $V(P) = V(G)$} \} \\
    \Prob{hamiltonian-cycle} &=& \{ \langle G \rangle \mid \text{ there exists a cycle $C\subseteq G$} \nonumber\\
    && \text{such that $V(C) = V(G)$} \} \\
    \Prob{clique} &=& \{ \langle G, k\rangle \mid \omega(G) = k \} \\
    \Prob{vertex-cover} &=& \{ \langle G, k\rangle \mid \beta(G) = k \} \\
    \Prob{bounded-degree-spanning-tree} &=& \{ \langle G, k\rangle \mid \text{there exists a $k$-bounded} \nonumber \\ 
    && \text{degree spanning tree in $G$} \} \\
    \Prob{clique-cover} &=& \{ \langle G, k\rangle \mid \text{$G$ has a clique cover of size} \nonumber \\ 
    && \text{at most $k$} \} \\
    \Prob{subgraph-isomorphism} &=& \{ \langle G, H\rangle \mid \text{$G$ contains a subgraph} \nonumber \\ 
    && \text{isomorphic to $H$} \}\\
    \Prob{travelling-salesperson} &=& \{ \langle G = (V, E), W : E \to \mathbb{R}, k\rangle \mid \text{$G$ is a weighted} \nonumber \\ 
    && \text{graph and $G$ contains a Hamiltonian cycle $C$} \nonumber \\ 
    && \text{such that $\sum_{e\in E(c)} W(e) \leq k$} \} 
\end{eqnarray}

\section{Main Results}\label{sec: main result}
The notion of boundary classes has been an active area of research since it was initially introduced by Vladimir E. Alekseev. One major question is if there exists any relation between existing boundary classes.
In this section, we introduce a result that shows a relation between boundary classes of problems that are reducible to each other. 
This relation gives us a method to obtain a new boundary class for a problem from existing known boundary classes of other problems.

For $\NPH$ problems $\Pi$ and $\Gamma$, $\Pi$ is said to be \emph{bi-reducible} to $\Gamma$ if there exists a polynomial time computable function $g: \Sigma^* \to \Sigma^*$ such that $\forall w \in \Sigma^*,\; w \in \Pi \iff g(w) \in \Gamma$ and there exists a subset $Z$ of $\Gamma$ such that the function $f$ obtained from restricting the codomain of $g$ to $Z$ is a bijection from $\Pi$ to $Z$ such that $\forall w \in \Sigma^*,\; w \in \Pi \iff f(w) \in Z$ and $f^{-1}$ is also polynomial time computable. The function $f$ is called a bi-reduction from $\Pi$ to $\Gamma$

If the problems $\Pi$ and $\Gamma$ are graph problems then for $f$, we will call the bijection $f_G$ a graph reduction under $f$ if a graph $G$ is mapped to a graph $G'$ under $f$ then $f_G(G) = G'$. So for the reduction $f$, $f_G$ is the graph to graph part of the mapping. So we remove the parameter mapping from $f$ and keep only the graph to graph map to obtain $f_G$.
If for every hereditary class $X$ of graphs, $f_G(X)$ is also hereditary then the graph reduction is called \emph{hereditary closed}.

\begin{theorem}\label{theorem3}
    Let $\Pi$ an $\Gamma$ be two $\NPH$ graph problems and let $f: \Sigma^* \to \Sigma^*$ be a bi-reduction from $\Pi$ to $\Gamma$ and the graph reduction $f_G$ is hereditary closed. Then a class of graph $X$ is a $(\Pi, A)$-boundary if and only if $f_G(X)$ is a $(\Gamma, f_G(A))$-boundary for some hereditary class $A \supseteq X$.
\end{theorem}

\begin{proof}
    Let $\Pi$ and $\Gamma$ be two $\NPH$ graph problems, let $f: \Sigma^* \to \Sigma^*$ be a bi-reduction from $\Pi$ to $\Gamma$ and $f_G$ be the graph reduction under $f$.

    First, we show that every boundary class $\Pi$ is also a boundary class of $\Gamma$. Let $X$ be a $(\Pi, A)$-boundary, then we know $X = \bigcap_{i \in \mathbb{Z}^+} X_i$, where each $X_i \supseteq X_{i+1}$ and each $X_i$ is $\Pi$-Hard. Let $X'_i = f_G(X_i)$, then $X'_i$ is $\Gamma$-Hard. Now let $X' = \bigcap_{i \in \mathbb{Z}^+} f_G(X_i) = \bigcap_{i \in \mathbb{Z}^+} X'_i$, then $X'$ is a limit class for $\Gamma$.
    Now we prove the minimality of $X'$. Suppose there exists $Y' \subset X'$ such that $Y'$ is a limit class for $\Gamma$. Then $Y' = \bigcap_{i \in \mathbb{Z}^+} Y'_i$ of $\Gamma$-Hard classes $Y'_i$ where $Y'_i \supseteq Y'_{i+1}$. Then $Y = f_G^{-1}(Y') = f_G^{-1}\left(\bigcap_{i \in \mathbb{Z}^+} Y'_i\right) \subset f_G^{-1}(X')$ so $Y \subset X$ and $Y = f_G^{-1}\left(\bigcap_{i \in \mathbb{Z}^+} Y'_i\right) = \bigcap_{i \in \mathbb{Z}^+} f_G^{-1}\left(Y'_i\right)$ and as $Y'_i$ is $\Gamma$-Hard, $f_G^{-1}\left(Y'_i\right)$ is $\Pi$-Hard. So now $Y$ is a limit class for $\Pi$, but $Y \subset X$ which is a contradiction as $X$ is a boundary class for $\Pi$.

    Conversely, suppose $Y$ is a $(\Gamma, f_G(A))$-boundary. Then $Y = \bigcap_{i \in \mathbb{Z}^+} Y_i$, where each $Y_i \supseteq Y_{i+1}$ and each $Y_i$ is $\Gamma$-Hard.
    Now let $Y'_i = f_G^{-1}(Y_i)$, then $Y'_i$ is $\Pi$-Hard. Let $Y' = \bigcap_{i \in \mathbb{Z}^+} f^{-1}_G(Y_i) = \bigcap_{i \in \mathbb{Z}^+} Y'_i$, then $Y'$ is a limit class for $\Pi$.
    Now we prove the minimality of $Y'$. Suppose there exists $X' \subset Y'$ such that $X'$ is a limit class for $\Pi$.
    Then $X' = \bigcap_{i \in \mathbb{Z}^+} X'_i$ of $\Pi$-Hard classes $X'_i$ where $X'_i \supseteq X'_{i+1}$.
    Then $X = f_G(X') = f_G\left(\bigcap_{i \in \mathbb{Z}^+} X'_i\right) \subset f_G(Y')$ so $X \subset Y$ and $X = f_G\left(\bigcap_{i \in \mathbb{Z}^+} X'_i\right) = \bigcap_{i \in \mathbb{Z}^+} f_G\left(X'_i\right)$ and as $X'_i$ is $\Pi$-Hard, $f_G\left(X'_i\right)$ is $\Gamma$-Hard.
    So now $X$ is a limit class for $\Gamma$, but $X \subset Y$ which is a contradiction as $Y$ is a boundary class for $\Gamma$.

    Therefore, for a hereditary class of graphs $X$, $X$ is a $(\Pi, A)$-boundary if and only if $f_G(X)$ is a $(\Gamma, f_G(A))$-boundary.

  \end{proof}

This result gives us a relationship between existing known boundary classes for $\NPH$ graph problems which are reducible to each other. We now use this result to obtain a few new boundary classes for $\NPH$ graph problems.

\subsection{New boundary classes}
We use our result given in Section~\ref{sec: main result} to find some new boundary classes for $\NPH$ graph problems.  We apply Theorem~\ref{theorem3} on all the problems defined in the Section~{\ref{sec:problems}}. This shows the strength of our result in furthering the study of boundary classes for $\NPH$ graph problems.

First, we introduce some graph classes we will use in our results. $T_{i,j,k}$ is a graph of the form illustrated in Figure~\ref{fig:T}. These are trees with at most three leaves. $T$ is the class of graphs in which every connected component is of the form $T_{i,j,k}$ for some $i, j,k \in \mathbb{N}$.

A \emph{caterpillar with hairs} of arbitrary length is a \emph{subcubic tree} in which all cubic vertices belong to a single path. Figure~\ref{fig:Q} shows the form of a caterpillar with hair of arbitrary length. $Q$ is the class of graphs in which every connected component is a caterpillar with hair of arbitrary length.

For a graph $G = (V, E)$ the transformation $L(G)$ is called the \emph{line graph} of $G$. The vertex set of the transformation $L(G)$ is $E$ (vertices of $G$ become the edges of $L(G)$), and two vertices $e$ and $f$ in $L(G)$ are adjacent in $L(G)$ if and only if corresponding edges $e$ and $f$ in $G$ are adjacent in $G$. Figure~\ref{fig:L} shows an example of the line graph transformation. For a class of graphs $X$, $L(X)$ is defined as $L(X) = \{L(G) \mid G \in X\}$. The inverse $L^{-1}(X)$

The transformation $R(G)$ is defined as replacing every cubic vertex of $G$ with a complete graph of three vertices,~$K_3$. Figure~\ref{fig: R} illustrates an example of an $R(G)$ transformation. For a class of graphs $X$, $R(X)$ is defined as $R(X) = \{R(G) \mid G \in X\}$.

We defined the complement $co(G)$ of a graph $G$ as a graph where the vertex set of $co(G)$ is the same as the vertex set of $G$, and for every pair of vertices $v, u \in V$, there is an edge between $u$ and $v$ in $co(G)$ if and only if there is no edge between $u$ and $v$ in $G$. Figure~\ref{fig:co} shows complement of $T_{1, 1, 1}$. For a class of graphs $X$, $co(X)$ is defined as $co(X) = \{co(G) \mid G \in X\}$.

The transformation $K(G)$ denotes the weighted complete graph of $|V|$ vertices obtained from $G$ such that for each edge $e$ in $E(K(G))$ the weight of $e$ is $0$ if $e$ is in $E(G)$ else the weight of $e$ is $1$. Figure~\ref{fig:K} shows an example of the $K(G)$ transformation. For a class of graphs $X$, $K(X) = \{K(G) \mid G\in X\}$

\begin{figure}[!tbh]
\centering
\resizebox{5cm}{5cm}{
\begin{tikzpicture}

    %T
    \node[style={fill=black,circle}] (0) at (4,0){};
    \node[style={fill=black,circle}] (1) at (4,1.5)[label=left:$1$]{};
    \node[style={fill=black,circle}] (2) at (4,3)[label=left:$2$]{};
    \node (3) at (4,4.5){$\vdots$}; 
    \node[style={fill=black,circle}] (4) at (4,6)[label=left:$i-1$]{};
    \node[style={fill=black,circle}] (5) at (4,7.5)[label=left:$i$]{};

    \node[style={fill=black,circle}] (6) at (2.7,-0.75)[label=below:$1$]{};
    \node[style={fill=black,circle}] (7) at (1.4,-1.5)[label=below:$2$]{};
    \node (8) at (0.102886,-2.25){$\udots$}; 
    \node[style={fill=black,circle}] (9) at (-1.196,-3)[label=below:$j-1$]{};
    \node[style={fill=black,circle}] (10) at (-2.495,-3.75)[label=below:$j$]{};

    \node[style={fill=black,circle}] (11) at (5.299,-0.75)[label=above:$1$]{};
    \node[style={fill=black,circle}] (12) at (6.598,-1.5)[label=above:$2$]{};
    \node (13) at (7.897,-2.25){$\ddots$}; 
    \node[style={fill=black,circle}] (14) at (9.196,-3)[label=above:$k-1$]{};
    \node[style={fill=black,circle}] (15) at (10.495,-3.75)[label=above:$k$]{};

    \draw[black,very thick] (0)--(1) (1)--(2) (2)--(3) (3)--(4) (4)--(5) (0)--(6) (6)--(7) (7)--(8) (8)--(9) (9)--(10)
    (0)--(11) (11)--(12) (12)--(13) (13)--(14) (14)--(15)
    ;

\end{tikzpicture}}   
\caption{A graph $T_{i,j,k}$ for some $i, j,k \in \mathbb{N}$}   
\label{fig:T}
\end{figure}
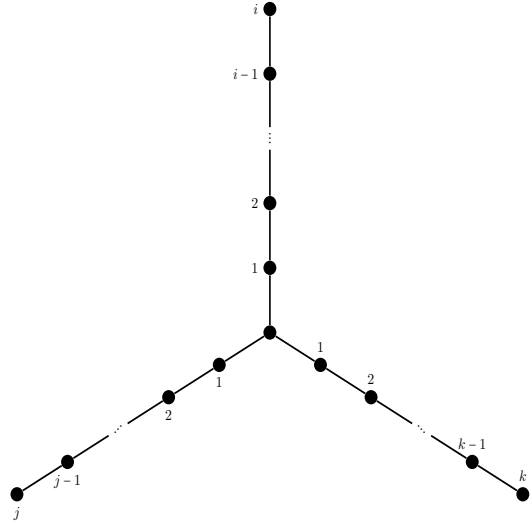

\begin{figure}[!tbh]
\centering
\resizebox{10cm}{3cm}{
\begin{tikzpicture}

    %Q
    \node (0) at (0,0){$\dots$}; 
    \node[style={fill=black,circle}] (1) at (4,0)[label=below:$1$]{};
    \node[style={fill=black,circle}] (2) at (8,0)[label=below:$1$]{};
    \node[style={fill=black,circle}] (3) at (12,0)[label=below:$1$]{};

    \node[style={fill=black,circle}] (4) at (4,1.5)[label=left:$2$]{};
    \node[style={fill=black,circle}] (5) at (8,1.5)[label=left:$2$]{};
    \node[style={fill=black,circle}] (6) at (12,1.5)[label=left:$2$]{};

    \node (7) at (4,3){$\vdots$}; 
    \node (8) at (8,3){$\vdots$}; 
    \node (9) at (12,3){$\vdots$}; 

    \node[style={fill=black,circle}] (10) at (4,4.5)[label=left:$i_{n-2}-1$]{};
    \node[style={fill=black,circle}] (11) at (8,4.5)[label=left:$i_{n-1}-1$]{};
    \node[style={fill=black,circle}] (12) at (12,4.5)[label=left:$i_n-1$]{};

    \node[style={fill=black,circle}] (13) at (4,6)[label=left:$i_{n-2}$]{};
    \node[style={fill=black,circle}] (14) at (8,6)[label=left:$i_{n-1}$]{};
    \node[style={fill=black,circle}] (15) at (12,6)[label=left:$i_n$]{};

    \node[style={fill=black,circle}] (-1) at (-4,0)[label=below:$1$]{};
    \node[style={fill=black,circle}] (-2) at (-8,0)[label=below:$1$]{};
    \node[style={fill=black,circle}] (-3) at (-12,0)[label=below:$1$]{};

    \node[style={fill=black,circle}] (-4) at (-4,1.5)[label=right:$2$]{};
    \node[style={fill=black,circle}] (-5) at (-8,1.5)[label=right:$2$]{};
    \node[style={fill=black,circle}] (-6) at (-12,1.5)[label=right:$2$]{};

    \node (-7) at (-4,3){$\vdots$}; 
    \node (-8) at (-8,3){$\vdots$}; 
    \node (-9) at (-12,3){$\vdots$}; 

    \node[style={fill=black,circle}] (-10) at (-4,4.5)[label=right:$i_{3}-1$]{};
    \node[style={fill=black,circle}] (-11) at (-8,4.5)[label=right:$i_{2}-1$]{};
    \node[style={fill=black,circle}] (-12) at (-12,4.5)[label=right:$i_{1}-1$]{};

    \node[style={fill=black,circle}] (-13) at (-4,6)[label=right:$i_{3}$]{};
    \node[style={fill=black,circle}] (-14) at (-8,6)[label=right:$i_{2}$]{};
    \node[style={fill=black,circle}] (-15) at (-12,6)[label=right:$i_{1}$]{};

    \draw[black,very thick] (0)--(1) (1)--(2) (2)--(3) (1)--(4) (2)--(5) (3)--(6)
    (7)--(4) (8)--(5) (9)--(6) (7)--(10) (8)--(11) (9)--(12) (13)--(10) (14)--(11) (15)--(12)
    (0)--(-1) (-1)--(-2) (-2)--(-3) (-1)--(-4) (-2)--(-5) (-3)--(-6)
    (-7)--(-4) (-8)--(-5) (-9)--(-6) (-7)--(-10) (-8)--(-11) (-9)--(-12) (-13)--(-10) (-14)--(-11) (-15)--(-12)
    ;

\end{tikzpicture}}     
\caption{A caterpillar with hair of arbitrary length for some $n \in \mathbb{N}$ and where each $i_k \in \mathbb{N}$.}   
\label{fig:Q}
\end{figure}
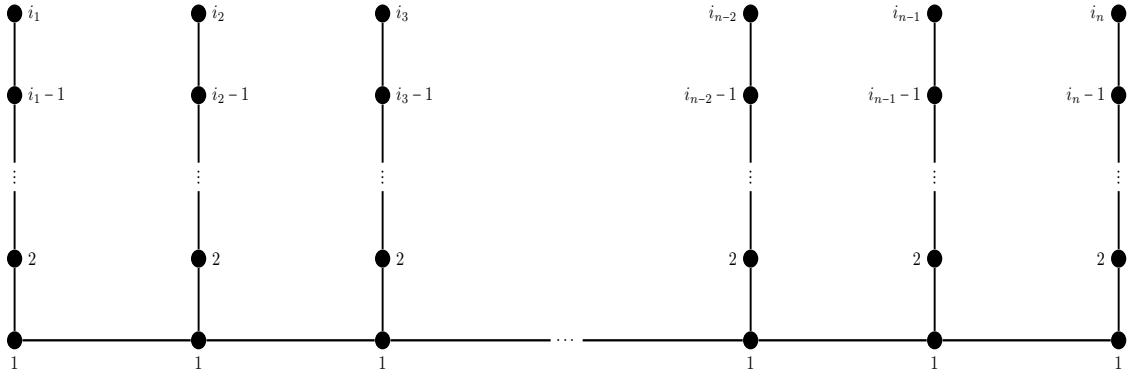

\begin{figure}[!tbh]
\centering
\resizebox{7cm}{3.5cm}{
\begin{tikzpicture}

    %L
    \node[style={fill=black,circle}] (0) at (4,0){};
    \node[style={fill=black,circle}] (1) at (4,1.5){};
    \node[style={fill=black,circle}] (2) at (4,3){};

    \node[style={fill=black,circle}] (3) at (2.7,-0.75){};
    \node[style={fill=black,circle}] (4) at (1.4,-1.5){};

    \node[style={fill=black,circle}] (5) at (5.299,-0.75){};
    \node[style={fill=black,circle}] (6) at (6.598,-1.5){};

    \draw[black,very thick] (0)--(1) (1)--(2) (0)--(3) (3)--(4) (0)--(5) (5)--(6);

    \node (a) at (8,2){\Large$\longrightarrow$}; 

    \node[style={fill=black,circle}] (7) at (12,1.5){};
    \node[style={fill=black,circle}] (8) at (12,3){};

    \node[style={fill=black,circle}] (9) at (10.7,-0.75){};
    \node[style={fill=black,circle}] (10) at (9.4,-1.5){};

    \node[style={fill=black,circle}] (11) at (13.299,-0.75){};
    \node[style={fill=black,circle}] (12) at (14.598,-1.5){};

    \draw[black,very thick] (7)--(8) (9)--(10) (11)--(12) (7)--(9) (9)--(11) (11)--(7);

\end{tikzpicture}   
}   
\caption{An example of the line graph transformation $L(G)$.}   
\label{fig:L}
\end{figure}
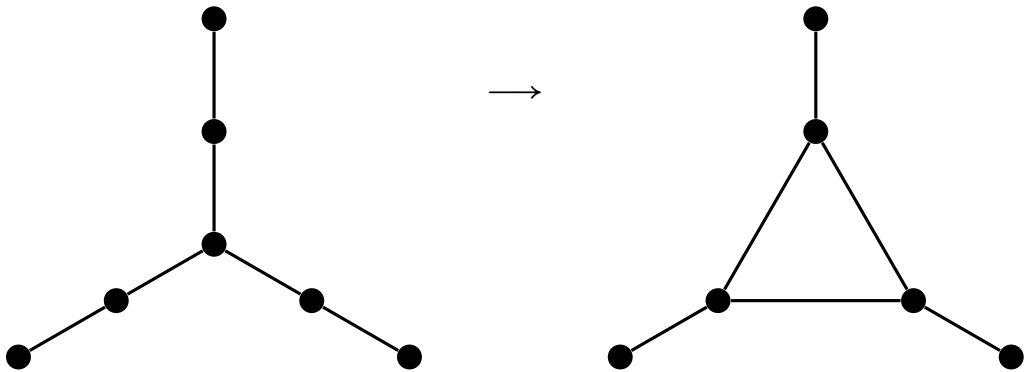

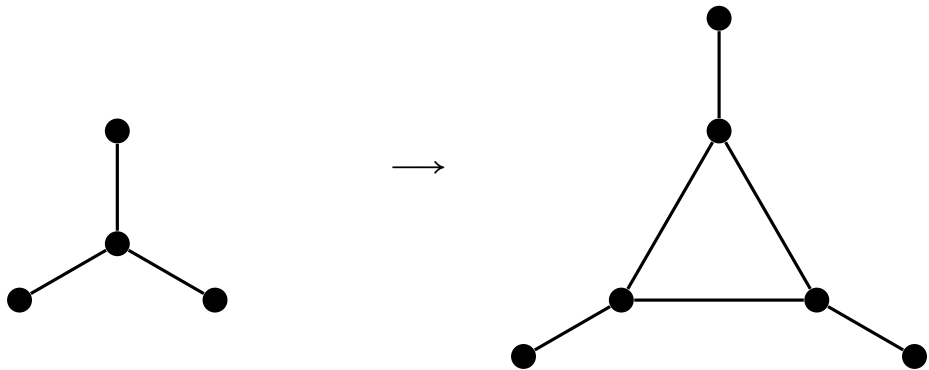
\begin{figure}[!tbh]
\centering
\resizebox{6cm}{3cm}{
\begin{tikzpicture}

    %L
    \node[style={fill=black,circle}] (0) at (4,0){};
    \node[style={fill=black,circle}] (1) at (4,1.5){};

    \node[style={fill=black,circle}] (3) at (2.7,-0.75){};

    \node[style={fill=black,circle}] (5) at (5.299,-0.75){};

    \draw[black,very thick] (0)--(1) (0)--(3) (0)--(5);

    \node (a) at (8,1){\Large$\longrightarrow$}; 

    \node[style={fill=black,circle}] (7) at (12,1.5){};
    \node[style={fill=black,circle}] (8) at (12,3){};

    \node[style={fill=black,circle}] (9) at (10.7,-0.75){};
    \node[style={fill=black,circle}] (10) at (9.4,-1.5){};

    \node[style={fill=black,circle}] (11) at (13.299,-0.75){};
    \node[style={fill=black,circle}] (12) at (14.598,-1.5){};

    \draw[black,very thick] (7)--(8) (9)--(10) (11)--(12) (7)--(9) (9)--(11) (11)--(7);

\end{tikzpicture}      
}
\caption{Example of transformation $R(G)$}   
\label{fig: R}
\end{figure}

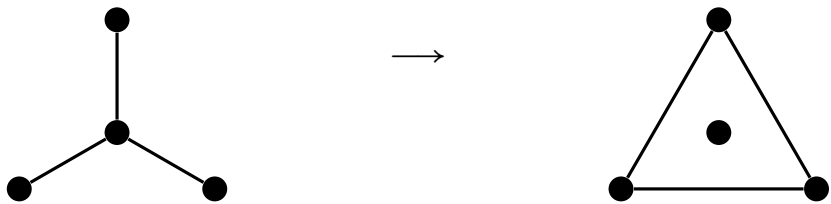
\begin{figure}[!tbh]
\centering
\resizebox{9cm}{3cm}{
\begin{tikzpicture}

    %co
    \node[style={fill=black,circle}] (0) at (4,0){};
    \node[style={fill=black,circle}] (1) at (4,1.5){};

    \node[style={fill=black,circle}] (3) at (2.7,-0.75){};

    \node[style={fill=black,circle}] (5) at (5.299,-0.75){};

    \draw[black,very thick] (0)--(1) (0)--(3) (0)--(5);

    \node (a) at (8,1){\Large$\longrightarrow$}; 

    \node[style={fill=black,circle}] (14) at (12,0){};
    \node[style={fill=black,circle}] (7) at (12,1.5){};

    \node[style={fill=black,circle}] (9) at (10.7,-0.75){};

    \node[style={fill=black,circle}] (11) at (13.299,-0.75){};

    \draw[black,very thick]  (7)--(9) (9)--(11) (11)--(7);

\end{tikzpicture}    
}
\caption{An example of the complement transformation $co(G)$}   
\label{fig:co}
\end{figure}

\begin{figure}[!tbh]
    \centering
    \resizebox{10cm}{3cm}{
    \begin{tikzpicture}
    
        %K
        \node[style={fill=black,circle}] (1) at (4,2.5){};
    \node[style={fill=black,circle}] (2) at (1.835,1.25){};
    \node[style={fill=black,circle}] (3) at (2.233,-1.767){};
    \node[style={fill=black,circle}] (4) at (5.767,-1.767){};
    \node[style={fill=black,circle}] (5) at (6.165,1.25){};
    \draw[black,very thick] 
    % (1)--(2) 
    (1)--(3) 
    (1)--(4) 
    % (1)--(5) 
    % (2)--(3) 
    (2)--(4) 
    (2)--(5) 
    % (3)--(4) 
    (3)--(5) 
    % (4)--(5)
    ;
    
    \node (arrow) at (10.5,1){\Large$\longrightarrow$}; 
    
    \node[style={fill=black,circle}] (a) at (17,2.5){};
    \node[style={fill=black,circle}] (b) at (14.835,1.25){};
    \node[style={fill=black,circle}] (c) at (15.233,-1.767){};
    \node[style={fill=black,circle}] (d) at (18.767,-1.767){};
    \node[style={fill=black,circle}] (e) at (19.165,1.25){};
    \draw[black,very thick] (a)--(b) node [midway,above] {1};
    \draw[black,very thick] (a)--(c) node [near end,above] {0};
    \draw[black,very thick] (a)--(d) node [near end, right] {0};
    \draw[black,very thick] (a)--(e) node [midway,above] {1};
    \draw[black,very thick] (b)--(c) node [midway,left] {1};
    \draw[black,very thick] (b)--(d) node [near end, below] {0};
    \draw[black,very thick] (b)--(e) node [midway,below] {0};
    \draw[black,very thick] (c)--(d) node [midway,below] {1};
    \draw[black,very thick] (c)--(e) node [near start,below] {0};
    \draw[black,very thick] (d)--(e) node [midway,right] {1};
        
    \end{tikzpicture}
    }   
    \caption{An example of the complete graph transformation $K(G)$}   
    \label{fig:K}
    \end{figure}
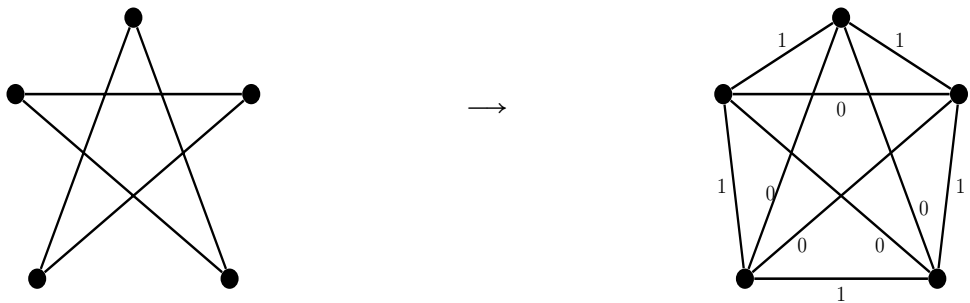

We now apply Theorem~\ref{theorem3} to obtain some new previously unknown boundary classes. 
In the following we will use the boundary class $T$ for the problem \Prob{independent-set} as defined by Alekseev in \cite{Alekseev03}. 
Nicholas~\etal showed in \cite{Lozin11} that both $Q$ and $R(Q)$ are boundary classes for the problem \Prob{hamiltonian-cycle}, and Munaro showed in \cite{Munaro17} that $Q$ is a boundary class for \Prob{hamiltonian-path}. 
In \cite{Malyshev11}, Malyshev showed that $co(L(T))$ is a boundary class for \Prob{vertex-coloring}.

\begin{corollary}\label{col1}
$T$ is a boundary class for \Prob{vertex-cover}.
\end{corollary}

\begin{proof}
Since $T$ is a boundary class for \Prob{independent-set}. For a graph $G = (V,E)$ we have that $\alpha(G) + \beta(G) = |V|$ (see \cite{West01}).
So we have that the function $f: \Sigma^* \to \Sigma^*$, such that $f(\langle G, k \rangle) = \langle G, |V|-k \rangle$ is a bi-reduction from \Prob{independent-set} to \Prob{vertex-cover} where $G$ is a graph and $k$ is a natural number. 
The graph reduction $f_G$ under $f$ is hereditary closed as the mapping maps every graph to itself. Therefore, by Theorem~\ref{theorem3}, $T$ is a boundary class for \Prob{vertex-cover}.
\end{proof}

\begin{corollary}\label{col2}
$co(T)$ is a boundary class for \Prob{clique}.
\end{corollary}

\begin{proof}
Since $T$ is a boundary class for \Prob{independent-set}. For a graph $G = (V,E)$ we have that $\alpha(G) = \omega(co(G))$ (see \cite{West01}). 
So we have that the function $f: \Sigma^* \to \Sigma^*$, such that $f(\langle G, k \rangle) = \langle co(G), k \rangle$ is a bi-reduction from \Prob{independent-set} to \Prob{clique} where $G$ is a graph, $k$ is a natural number and $co(G)$ is the complement graph of $G$.
The graph reduction $f_G$ under $f$ is hereditary closed as, if we remove a vertex $v$ from $co(G) \in co(T)$ then the graph $co(G)-v$ obtained is the image of $(G-v)\in T$ under $f_G$ and so $(co(G)-v) \in co(T)$.
Therefore, by Theorem~\ref{theorem3}, $co(T)$ is a boundary class for \Prob{clique}.
\end{proof}

\begin{corollary}\label{col3}
$co(T)$ is a boundary class for \Prob{subgraph-isomorphism}.
\end{corollary}

\begin{proof}
    From Corollary~\ref{col2} we know that $co(T)$ is a boundary class for \Prob{clique}. From \cite{Garey90} we have that $\Prob{clique} \leq_p \Prob{subgraph-isomorphism}$. 
    So we have that the function $f: \Sigma^* \to \Sigma^*$, such that $f(\langle G, n \rangle) = \langle G, K_n \rangle$ is a bi-reduction from \Prob{clique} to \Prob{subgraph-isomorphism}where $G$ is a graph, $n$ is a natural number and $K_n$ is the complete graph of $n$ vertices. 
    The graph reduction $f_G$ under $f$ is hereditary closed as the mapping maps every graph to itself. Therefore, by Theorem~\ref{theorem3}, $co(T)$ is a boundary class for \Prob{subgraph-isomorphism}.
\end{proof}

\begin{corollary}\label{col4}
    $K(Q)$ is a boundary class for \Prob{traveling-salesperson}.
\end{corollary}

\begin{proof}
    Since $Q$ is a boundary class for the \Prob{hamiltonian-cycle}. It is well known that $\Prob{hamiltonian-cycle} \leq_p \Prob{traveling-salesperson}$. 
    So we have that the function $f: \Sigma^* \to \Sigma^*$, such that $f(\langle G \rangle) = \langle K(G), 0 \rangle$ is a bi-reduction from \Prob{hamiltonian-cycle} to \Prob{traveling-salesperson} where $G$ is a graph, and $K(G)$ is the graph obtained from $G$ by $K(\cdot)$ transformation. 
    The graph reduction $f_G$ under $f$ is hereditary closed as, if we remove a vertex $v$ from $K(G) \in K(Q)$ then the graph $(K(G)-v)$ obtained is the image of $(G-v)\in Q$ under $f_G$ and so $(K(G)-v) \in K(Q)$. Therefore, by Theorem~\ref{theorem3}, $K(Q)$ is a boundary class for \Prob{traveling-salesperson}.
\end{proof}

\begin{corollary}\label{col5}
    $K(R(Q))$ is a boundary class for \Prob{traveling-salesperson}.
\end{corollary}

\begin{proof}
    Since $R(Q)$ is also a boundary class for the \Prob{hamiltonian-cycle} and it is well known that 
    \\$\Prob{hamiltonian-cycle} \leq_p \Prob{traveling-saleperson}$. 
    So we have that the function $f: \Sigma^* \to \Sigma^*$, such that $f(\langle G \rangle) = \langle K(G), 0 \rangle$ is a bi-reduction from \Prob{hamiltonian-cycle} to \Prob{traveling-saleperson} where $G$ is a graph, and $K(G)$ is the graph obtained from $G$ by the $K(\cdot)$ transformation. 
    The graph reduction $f_G$ under $f$ is hereditary closed as, if we remove a vertex $v$ from $K(G) \in K(R(Q))$ then the graph $(K(G)-v)$ obtained is the image of $(G-v) \in R(Q)$ under $f_G$ and so $(K(G)-v) \in K(R(Q))$. Therefore, by Theorem~\ref{theorem3}, $K(Q)$ is a boundary class for \Prob{traveling-saleperson}.
\end{proof}

\begin{corollary}\label{col6}
    $L(T)$ is a boundary class for \Prob{clique-cover}.
  \end{corollary}
  \begin{proof}
    Since $co(L(T))$ is a boundary class for \Prob{vertex-coloring}. 
    Karp showed in his 1972 paper that $\Prob{vertex-coloring} \leq_p \Prob{clique-cover}$ \cite{Karp72}. 
    For a graph $G = (V, E)$ we have that $G$ is $k$-colorable if and only if $co(G)$ has a clique cover of size $k$.
    So we have that the function $f: \Sigma^* \to \Sigma^*$, such that $f(\langle G, k \rangle) = \langle co(G), k \rangle$ is a bi-reduction from \Prob{vertex-coloring} to \Prob{clique-cover} where $G$ is a graph, $k$ is a natural number and $co(G)$ is the complement graph of $G$.
    The graph reduction $f_G$ under $f$ is hereditary closed as, if we remove a vertex $v$ from $co(G) \in L(T)$ then the graph $co(G)-v$ obtained is the image of $G-v \in co(L(T))$ under $f_G$ and so $co(G)-v \in L(T)$.
    Therefore, by Theorem~\ref{theorem3}, $L(T)$ is a boundary class for \Prob{clique-cover}.
  \end{proof}

  \begin{corollary}\label{col7}
    $Q$ is a boundary class for \Prob{bounded-degree-spanning-tree}.
  \end{corollary}
  \begin{proof}
    Since $Q$ is a boundary class for \Prob{hamiltonian-path}. 
    From \cite{Garey90} we know that $\Prob{hamiltonian path} \leq_p \Prob{bounded-degree-spanning-tree}$. For a graph $G = (V, E)$, $G$ has a Hamiltonian path if and only if $G$ has a spanning tree of bounded degree of 2.
    So we have that the function $f: \Sigma^* \to \Sigma^*$, such that $f(\langle G \rangle) = \langle G, 2 \rangle$ is a bi-reduction from \Prob{hamiltonian-path} to \Prob{bounded-degree-spanning-tree} where $G$ is a graph.
    The graph reduction $f_G$ under $f$ is hereditary closed as the mapping maps every graph to itself.
    Therefore, by Theorem~\ref{theorem3}, $Q$ is a boundary class for \Prob{bounded-degree-spanning-tree}.
  \end{proof}

  Through the application of our main theorem, we obtained seven new previously unknown boundary classes. With our result we able to get new boundary classes from previously known boundary classes and reductions. This shows us the power of our results and also gives us new boundary classes.

\section{Conclusion}\label{sec: conlusion}
The notion of boundary classes is a helpful tool in studying the computational complexity of $\NPH$ graph problems.
In this paper, we introduced a method for transforming the boundary class of one $\NPH$ graph problem into a boundary class for another $\NPH$ graph problem. 
This method gave us a relationship between different boundary classes of different $\NPH$ graph problems. Furthermore, it also gave us a relationship between boundary classes of $\NPH$ graph problems and a polynomial time reduction between them. This result is helpful to further the study of boundary classes of $\NPH$ graph problems.
We use our method to obtain seven new boundary classes for \emph{vertex cover}, \emph{clique}, \emph{traveling-salesperson}, \emph{bounded-degree-spanning-tree}, \emph{clique-cover} and \emph{subgraph-isomorphism}.
For all of these problems we provided the first boundary class. This shows the strength of our main result.
Our method can be used in the future to obtain more boundary classes for $\NPH$ graph problems. If we can obtain a reduction between $\NPH$ graph problems with the sufficient and necessary conditions we introduced then we can use those reductions to obtain further boundary classes.
Given that all problems in $\NPC$ are reducible to each other we can use our method on many $\NPH$ graph problems. 
Proving a class to be the boundary class for a $\NPH$ graph problem isn't a simple task, however, our theorem gives a method to obtain new previously unknown boundary classes for $\NPH$ graph problems.
This gives us a powerful tool to study the computational complexity of $\NPH$ graph problems and obtain their boundary classes.

\bibliographystyle{plain} 
\bibliography{reducibilityBoundaryClass}

\end{document}